\documentclass[nonacm,sigplan,screen]{acmart}
\usepackage{graphicx} 
\usepackage{cleveref}
\usepackage{enumitem}
\usepackage{cancel}
\usepackage{stmaryrd}
\usepackage{adjustbox}
\usepackage{cprotect}
\usepackage{fancyvrb}
\usepackage{xspace}
\usepackage{mathpartir}
\usepackage[dvipsnames]{xcolor}
\usepackage{xypic}
\usepackage{listings}
\usepackage{enumitem}
\usepackage{algorithm}
\usepackage{algpseudocode}
\usepackage{balance}

\newtheorem{theorem}{Theorem}
\newtheorem{assumption}[theorem]{Assumption}
\newtheorem{lemma}[theorem]{Lemma}

\setlist[itemize]{leftmargin=1.5em}

\lstdefinelanguage{vega}{
    morekeywords={transform,ops,fields,type,groupby,as,expr,\$schema,title,width,height,data,name,
      range,domain,field,orient,scale,from,encode,enter,value,band,url,source,scales,axes,marks,x,y,x2,update,exit},
    morekeywords=[2]{undefined,null,NaN,true},
    morekeywords=[3]{aggregate,filter,formula},
    sensitive=true,
    morestring=[b]",
    morecomment=[l]{//}
}
\lstdefinestyle{vegaStyle}{
    language=vega,
    basicstyle=\linespread{0.9}\sffamily,
    numbers=none,
    xleftmargin=1em,
    columns=fullflexible,
    frame=none,
    showstringspaces=false,
    keywordstyle=\color{MidnightBlue},
    keywordstyle=[2]\color{Sepia},
    keywordstyle=[3]\color{RedViolet},
    stringstyle=\color{Bittersweet},
    commentstyle=\color{Gray}\itshape,
    literate=
      {datum}{{{\color{Sepia}{datum}}}}1
      {x2}{{{\color{MidnightBlue}{x2}}}}1
      {"aggregate"}{{{\color{RedViolet}{"aggregate"}}}}1
      {"filter"}{{{\color{RedViolet}{"filter"}}}}1
      {"formula"}{{{\color{RedViolet}{"formula"}}}}1
      {0}{{{\color{OliveGreen}{0}}}}1
      {1}{{{\color{OliveGreen}{1}}}}1
      {2}{{{\color{OliveGreen}{2}}}}1
      {4}{{{\color{OliveGreen}{4}}}}1
      {9}{{{\color{OliveGreen}{9}}}}1
      {:}{{:}}1
}

\def\uv#1{``#1''}

\newcommand{\de}{\mathrel{\mathop{\Coloneqq}\ \ }}
\newcommand{\densp}{\mathrel{\mathop{\Coloneqq}}}
\newcommand{\orr}{\mathbin{\mathop{|}}}
\newcommand{\out}{{\rm out}}
\newcommand{\new}{{\rm new}}
\newcommand{\inn}{{\rm in}}
\newcommand{\init}{{\rm init}}
\newcommand{\prim}{{\rm primitive}}

\newcommand{\spec}{\mathsf{spec}}

\newcommand{\truthy}{\mathsf{tru}}
\newcommand{\falsy}{\overline{\truthy}}

\newcommand{\bool}{\mathsf{boolean}}
\newcommand{\num}{\mathsf{number}}

\newcommand{\str}{\mathsf{string}}

\newcommand{\ls}{\mathsf{ls}}

\newcommand{\groupOf}{\mathsf{groupOf}}
\newcommand{\formula}{\mathsf{\color{RedViolet}formula}}
\renewcommand{\state}{\mathsf{state}}

\newcommand{\cross}{\mathsf{\color{RedViolet} cross}}

\newcommand{\inputt}{\mathsf{input}}
\newcommand{\error}{\mathsf{error}}
\newcommand{\add}{\mathsf{add}}
\newcommand{\gs}{\mathsf{cs}}

\newcommand{\remove}{\mathsf{remove}}
\newcommand{\update}{\mathsf{update}}

\newcommand{\filter}{\mathsf{\color{RedViolet}filter}}
\newcommand{\data}{\mathsf{\color{MidnightBlue} data}}
\newcommand{\round}{\mathsf{\color{Bittersweet} round}}
\newcommand{\contains}{\mathsf{\color{Bittersweet} round}}

\newcommand{\op}{\mathsf{op}}

\newcommand{\summ}{\mathsf{\color{MidnightBlue} sum\ }}
\newcommand{\avg}{\mathsf{\color{MidnightBlue} avg\ }}
\newcommand{\sumi}{\mathbf{sum}}

\newcommand{\DataChange}{\mathsf{DataChange}}

\newcommand{\countt}{\mathsf{\color{MidnightBlue} count}}
\newcommand{\counta}{\mathsf{count}}
\newcommand{\tto}{\longrightarrow}
\newcommand{\agg}{\mathsf{\color{RedViolet} agg}}
\newcommand{\aggb}{\mathsf{agg}}

\newcommand{\as}{\mathbin{\mathop{\mathsf{\color{MidnightBlue} as}}}}

\newcommand{\lookup}{\mathsf{lookup}}
\newcommand{\buildT}{\mathsf{buildT}}
\newcommand{\buildD}{\mathsf{buildD}}
\newcommand{\buildS}{\mathsf{buildS}}
\newcommand{\datum}{\mathrm{\color{Sepia} datum}}
\newcommand{\invals}{\mathsf{invals}}

\newcommand{\tinyvega}{\textsc{$\upmu$Vega}\xspace}

\newcommand{\vinl}[1]{\lstinline[style=vegaStyle]{#1}}
\newcommand{\vkvd}[1]{\textsf{\color{MidnightBlue} #1}}
\newcommand{\vident}[1]{\textsf{\color{Bittersweet} #1}}
\newcommand{\gnum}[1]{{\color{OliveGreen} #1}}

\begin{document}
\title[Formal Semantics and Type System for Vega Data Transformations]{Formal Semantics and Type System for\\ Vega Data Transformations}
\author{Kristýna Petrlíková}
\affiliation{%
  \institution{Charles University}
  \city{Prague, Czech Republic}\country{}
}
\email{auburn@kam.mff.cuni.cz}

\author{Tomas Petricek}
\affiliation{%
  \institution{Charles University}
  \city{Prague, Czech Republic}\country{}
}
\email{tomas@tomasp.net}


\begin{abstract}


Vega is a popular declarative language for creating interactive data visualizations.
It supports reactive data transformations using its streaming
dataflow architecture. Despite its widespread adoption, the exact semantics of Vega
is subtle and poorly documented.
%
This leads to incorrect or confusing visualizations and difficult-to-understand error messages.

This paper makes two contributions. First, we
define a graph-based operational semantics, providing a precise model of the streaming
dataflow architecture of Vega. Second, we present a~type system for the core data
transformation language of Vega, which can prevent a range of common errors. 
We show that our type system is sound with respect
to the semantics. 

%
%
%
While the dataflow architecture of Vega closely resembles well-studied models such as functional reactive programming and adaptive computation, there are important differences. The novelty of our work lies in making these precise and providing static analysis for such a reactive data visualization language. The result is a checker for Vega that can catch common real-world errors.
\end{abstract}


\maketitle

\section{Introduction}
\label{sec:intro}

Vega \cite{satyanarayan-2014-vega} is a declarative language for creating data visualizations inspired by the grammar of graphics~\cite{wilkinson-2012-gg}.
Vega makes it possible to define rich interactive visualizations, but the exact working of
the system is complex and poorly documented. 

First, Vega is built on top of JavaScript and
inherits some of its subtle semantics. Moreover, the implementation uses many defaults that might appear unexpected or unintuitive. For example, the default output fields for the \texttt{stack} transform are \uv{\texttt{y0}} and \uv{\texttt{y1}}, whereas for the \texttt{bin} transform they are \uv{\texttt{a}} and \uv{\texttt{b}} \cite{vega-transforms-docs}.

Second, Vega uses a~streaming architecture
\cite{reactive-vega-architecture-2016}, which makes understanding dataflow challenging.
As a result, debugging Vega visualizations is a known challenge \cite{vega-debugging-2016}.


\subsubsection*{Debugging Vega Errors}
Consider the visualization in \autoref{fig:books-ok}, which takes a collection of books and
aggregates the data to show the average word count per book. The transformations are specified
as follows: 

\begin{lstlisting}[style=vegaStyle]
transform: [
  { type: "aggregate", groupby: ["classification"],
   ops: ["mean", "count"], fields: ["words", ""],
   as: ["avg_words", "book_count"] },
  { type: "filter", "expr": "datum.book_count >= 10" },
  { type: "formula", as: "avg_hundreds",
   expr: "datum.avg_words / 100" }
]
\end{lstlisting}

\begin{figure}[t]
  \centering
  \vspace{1em}
  \includegraphics[width=0.5\textwidth]{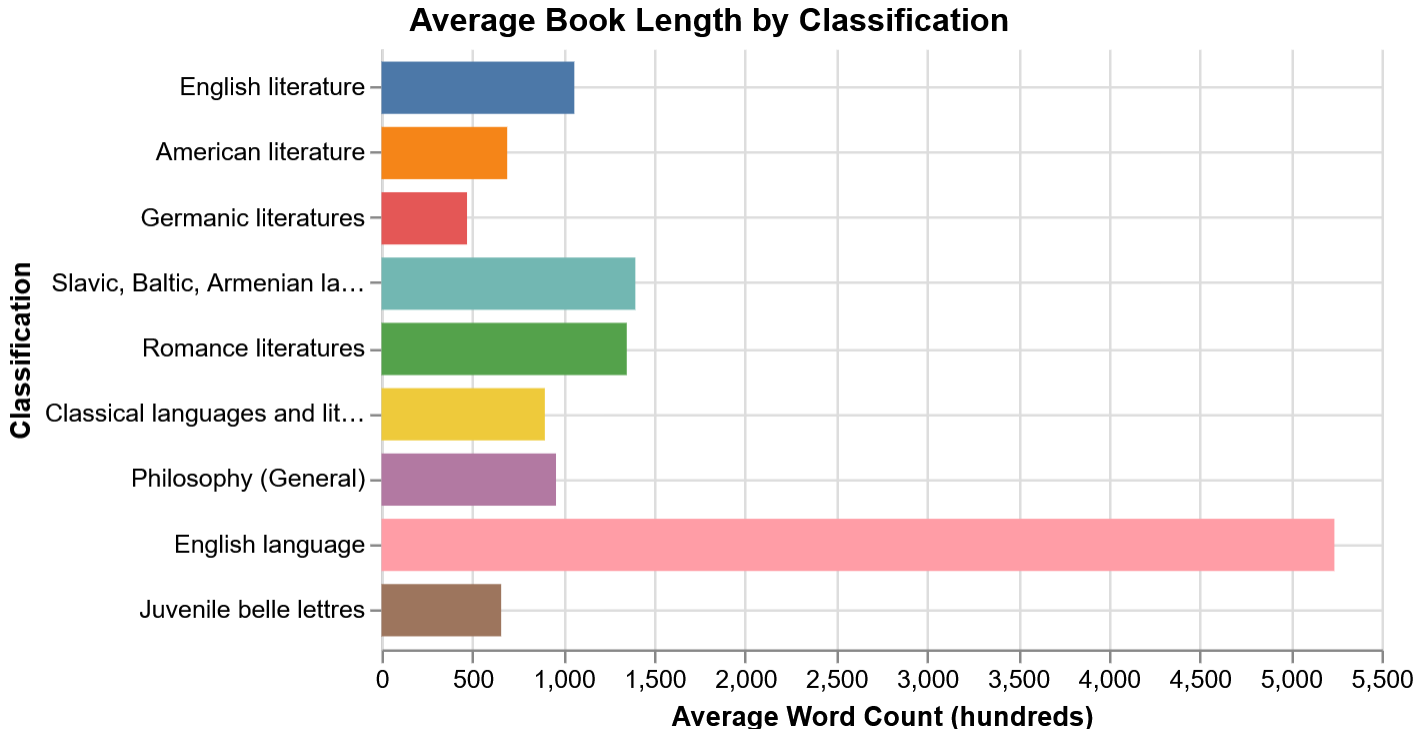}
  \caption{Vega visualization of average book lengths.}
  \label{fig:books-ok}
\end{figure}


The example first uses the \vinl{aggregate} transform to group books by their
\vident{classification} and computes an average word count and a number of
books in each group.
It then adds the word count in hundreds as \vident{avg\_hundreds}
and removes groups with fewer than 10 books.

Let's say that we accidentally type \vident{avgwords} instead of \vident{avg\_words}
in the \vinl{formula} transform:

\begin{lstlisting}[style=vegaStyle]
{ type: "formula", as: "avg_hundreds",
 expr: "datum.avgwords / 100" }
\end{lstlisting}

Vega displays an empty visualization and reports a warning 
``Infinite extent for field avg\_hundreds\_start''. The reason is that accessing a non-existent
field in JavaScript returns \vinl{undefined} and \vinl{undefined/100} evaluates to
\vinl{NaN}. Vega thus tries to draw a chart where all bars have the same \vinl{NaN} length
and fails to calculate the scale for the axis.

Now, let's say that we try to change the filter, after the aggregation, to
look only at non-English books:

\begin{lstlisting}[style=vegaStyle]
{ type: "filter", "expr": "datum.language != 'en'" }
\end{lstlisting}

Vega displays a visualization that includes all groups, counting all the books in the input
dataset. The reason is that the earlier \vinl{aggregate} operation produces values that contain
only the group key (\vident{classification}) and newly computed aggregations (\vident{avg\_words}
and \vident{book\_count}). The field \vident{language} is no longer present and
\vinl{undefined != "en"} evaluates to \vkvd{true}.

The root cause of both issues is access to an invalid field, but they are exhibited in
different ways that are easy to overlook. Moreover, the exact source of the error is not 
easy to trace because of the streaming dataflow architecture~of~Vega.

\subsubsection*{Understanding Vega Architecture}
Vega makes it possible to create interactive data visualizations, in which user interface events such as moving the mouse cursor
can cause items in the dataset to be added, updated or removed. To support this,
the Vega runtime \cite{reactive-vega-architecture-2016} builds a directed acyclic dataflow graph
that represents dependencies between components. The nodes (called \emph{operators} in Vega) can
have an internal state
. The state is used, e.g., by aggregation to keep track of values
encountered so far. The source nodes (those without incoming edges) correspond to data sources or
user interface events.

When a dataset is created or an event occurs, the source node creates a \emph{changeset}
containing a set of newly created, updated, and removed values. The changeset then \emph{pulses} (propagates)
along the graph in topological order and is transformed in various ways by the individual nodes. 
For example, the \vinl{filter} transform marks updated values that
no longer satisfy the given condition as removed.

The streaming dataflow architecture makes error reporting more challenging. It also makes
defining accurate semantics of Vega an interesting theoretical challenge.

\begin{figure}
\begin{lstlisting}[style=vegaStyle]
{ $schema: "https://vega.github.io/schema/vega/v5.json",
 title: "Average Book Length by Classification",
 width: 400, height: 240,
 data: [ // Define data sources and transformations
  { name: "books", url: "data/books.json" },
  { name: "aggs", source: "books", transform: [ ... ] }
 ],
 scales: [ ... ], // Scales map data values to visual values
 axes: [ ... ],   // Axes show scales in the visualization
 marks: [ ... ]   // Define visual marks in the visualization
};
\end{lstlisting}
\vspace{-1em}
\caption{The structure of a Vega specification.}
\label{fig:spec}
\vspace{-1em}
\end{figure}

\subsubsection*{Related Work}
Vega is well-studied from the user-centric perspective
\cite{satyanarayan-2014-vega,vega-lite-2017}, but the streaming
dataflow architecture of Vega has only been described in high-level terms
\cite{reactive-vega-architecture-2016}. The execution model is akin to event-driven
functional reactive programming \cite{elliott-1997-fran,WanEFRP} and adaptive computation
\cite{acar-2002-adaptive}, which have been studied in more depth.

The semantics of reactive programming has been described both in operational
\cite{ischard-2025-frp} and denotational style \cite{elliott-2009-pushpull,krishnaswami-2011-ultrametric}.
Oeyen \cite{oeyen-2023-graph} presents a graph-based model, which is close to how the Vega runtime
operates, but the model works with singular values, rather than updates to a dataset.

Verifying the properties of Vega visualizations has been done through linters \cite{chen2021,mcnutt2018linting},
but these focus mainly on correcting misleading visualizations. Vega errors can also be
detected using debugging and profiling tools \cite{yang-2023-profiling,vega-debugging-2016}.
Static checking of specifications is an under-explored area.

\subsubsection*{Contributions} We present a formal model and a type system for
Vega data transformations. Our contributions are:
\begin{itemize}[itemsep=3pt]
\item We present the \tinyvega calculus, which captures the essence
  of Vega data transformations (\autoref{sec:essenc}), and formally model its evaluation using
  a graph-based operational semantics based on changeset propagation (\autoref{sec:sem}).

\item We define a type system for \tinyvega (\autoref{sec:types}) that ensures the
  correctness of Vega specifications, and shows its soundness with respect to our
  operational semantics (\autoref{sec:sound}).

\item We introduce a CLI-based Vega specification checker that can check a subset of real-world Vega
  specifications for common errors such as invalid field accesses (\autoref{sec:checker}), and which accounts for Vega's unexpected defaults.

%

\end{itemize}

\section{Vega Overview}
A data visualization in Vega is defined by a JSON specification.
\autoref{fig:spec} shows the structure of the specification of the data visualization discussed
in \autoref{sec:intro}. Although the focus of this paper is on data transformations, we give
a~brief overview of the remaining components here. In the second part of this section, we discuss the properties of the runtime.

\subsection{Specification}
The most important parts of a Vega specification are sections that define the data sources
and transformations (\vinl{data}), mapping of data to visual attributes (\vinl{scale})
and the visual elements of the chart (\vinl{axes}, \vinl{marks}).

\subsubsection*{Data Sources and Transformations}
The \vinl{data} block in \autoref{fig:spec} defines two datasets. The first, \vident{books},
reads data from a specified JSON file. The second, \vident{aggs}, is obtained
from the \vinl{source} dataset \vident{books} through a series of data \textit{transformations} discussed
in \autoref{sec:intro}. Besides these, Vega offers a plethora of built-in transformations (total of 51 at the time of
writing) that reshape and filter data, calculate new fields, or derive new data streams
\cite{vega-transforms-docs}. We discuss the different kinds of transformations in \autoref{sec:essenc}.
It is worth noting that dataset definitions in Vega are processed sequentially and Vega reports
an error if the \vinl{source} dataset has not yet been defined.

\subsubsection*{Scales and Axes}
Vega separates the mapping of values from the data domain (e.g., classification,
numbers of words) to the visual domain (pixels on screen) from the specification of visual elements.
The former is done in the \vinl{scales} block.

\autoref{fig:scales-axes-marks} shows two scale definitions for our running example. The
\vident{yscale} is a scale of type \vident{band}, which defines mapping from a discrete domain
(string values) to a continuous numerical range (pixels). The \vinl{range} of the scale
sets the maximal value of the target domain to the height of the visualization. Finally,
\vident{domain} specifies the source of the data for the scale. This is the \vident{classification}
field of the data in the computed dataset named \vident{aggs}. The \vident{xscale} is
a linear scale mapping a continuous range of values (word counts) to a numerical range
defined by the width of the visualization. This time, the source uses the computed
\vident{avg\_hundreds} field.

\begin{figure}
\begin{lstlisting}[style=vegaStyle]
scales: [
  { name: "yscale", type: "band", range: "height",
   domain: { data: "aggs", field: "classification" } },
  { name: "xscale", type: "linear", range: "width",
   domain: { data: "aggs", field: "avg_hundreds" } }
],
axes: [
  { orient: "left", scale: "yscale", title: "Classification" },
  { orient: "bottom", scale: "xscale", title: "Avg Words" }
],
marks: [ {
  type: "rect", from: { data: "aggs" },
    encode: { enter: {
      x: { scale: "xscale", value: 0 },
      x2: { scale: "xscale", field: "avg_hundreds" },
      y: { scale: "yscale", field: "classification" },
      height: { scale: "yscale", band: true },
    } }
} ]
\end{lstlisting}
\vspace{-0.5em}
\caption{Scales, axes and marks for our running example.}
\vspace{-0.5em}
\label{fig:scales-axes-marks}
\end{figure}

Axes are visual elements used for displaying information about scales. In
\autoref{fig:scales-axes-marks}, we specify that axes should appear on the left and bottom,
we define the \vinl{title} for each of the axes and we specify what scales they represent.

\subsubsection*{Marks}
The content of a chart is defined in \vinl{marks}. Marks encode data from some data source as
geometric elements such as rectangles. Vega includes many types of marks, ranging from lines,
paths and rectangles to icons and text.

Our example, shown in \autoref{fig:scales-axes-marks}, defines one set of marks whose \vinl{type}
is \vident{rect}. The data source, specified using \vinl{from}, is the computed dataset
\vident{aggs}. In Vega, the source for a mark can be any dataset, but also another mark.
This feature, referred to as \emph{reactive geometry}\cite{VegaRG} makes it possible to specify new set
of marks relative to existing previously defined marks, for example to add text labels to
the visual display. The \vinl{from} definition can also specify faceting where a single dataset
is split across multiple groups of marks.

Each mark has a set of encoding rules that map the backing data to the visual properties of
the mark. There is a number of different encoding rules, applied in different stages of a
mark's lifetime. The three basic encoding rules cover the cases when a mark is first created
(\vinl{enter}), when data is updated (\vinl{update}) and when it is removed (\vinl{exit}).
In our example, we only specify the \vinl{enter} rule for creating the marks.

Depending on the type of the mark, the encoding rule needs to define properties that are
used for drawing the mark. In case of a horizontal rectangle, we define its range using
a pair of X coordinates (\vinl{x} and \vinl{x2}), Y offset (\vinl{y}) and its height
(\vinl{height}). The object that specifies the value, called \emph{value reference} in
Vega, can be a constant (\vinl{value}), reference to a field (\vinl{field}) as well
as a signal or a color palette reference (not discussed here). Optionally, a specified
\vinl{scale} can be used to transform domain values to the visual domain. Finally, the
\vinl{band} field can be used to refer to the size of the band computed by the scale of type
\vident{band}.

\begin{figure*}[t]
\vspace{-1em}
\begin{alignat*}{3}
		\text{Field names}\quad&&f&~\de \texttt{String}\\
		\text{Primitive values}\quad&&\mathsf{v}&~\de b \in \texttt{Boolean} \orr i \in \texttt{Number} \orr s \in \texttt{String} \\
    \text{Path access}\quad&&p&~\de \datum \orr p.f\\
		\text{Simple expressions}\quad&&a&~\de\mathsf{v} \orr p \orr a_1 + a_2\orr a_1 - a_2 \orr a_1 \times a_2 \orr a_1 \mathop{/} a_2 \orr \mathit{builtin}(e_1, \ldots, e_n)\\
		\text{General expressions}\quad&& e&~\de a \orr[e_1,\ldots,e_n] \orr\{f_1\colon e_1, \ldots, f_n\colon e_n\} \\
		\text{Aggregation operators}\quad&& \op &~\de \summ p_{\inn} \as f_{\out} \orr \avg p_{\inn} \as f_{\out} \orr \countt\as f_{\out}\\
		\text{Transforms}\quad&& t &~\de \agg( p_1,\ldots,p_n,\op_1,\ldots,\op_k)
		    \orr\cross(f_a, f_b) \orr\formula(e, f_{\out}) \orr\filter(e)\\
		\text{Dataset identifiers}\quad&& n &~\de \texttt{String} \\
		\text{Dataset definitions}\quad&& d  &~\de \data(n_{\out}, n_{\inn}, t_1, t_2, \ldots, t_k)\\
    \hspace{7em} 
    \text{Specifications}\quad&& \spec &~\de d_1, \ldots, d_k
\end{alignat*}
\vspace{-1em}
	\caption{Syntax of \tinyvega, which captures the essence of Vega.}
	\label{fig:syntax-new}
\end{figure*}

\subsubsection*{Event Streams and Signals} 
Interactive data visualizations can be defined using event streams and signals. Event streams model primitive DOM events such as \texttt{click} or \texttt{mouseover}. More complex event streams may be formed by filtering or merging existing ones.

Signals are dynamic variables that can be used to make a value in the specification time-varying. They are defined by an initial value and a set of event handlers that update the signal's value. A signal with event handlers is reactive and re-evaluates as events occur, whereas one with only an initial value is non-reactive; the two correspond to reactive and non-reactive behaviors in E-FRP \cite{WanEFRP}. Signals cannot be used for dataset, signal, or scale names, which must be fixed in their definitions.
Signals can also be composed to form complex dependency graphs. While there is no ordering requirement for signals (as opposed to dataset definitions), Vega requires the resulting graph to be acyclic.

We focus on modeling the streaming dataflow that propagates dataset updates through data transformations. Event streams and signals would be an interesting addition as they make Vega more dynamic, akin to functional reactive programming frameworks. 

\subsection{Runtime}
\label{subsec:runtime}
After the specification has been parsed and validated using JSON Schema \cite{JSONSchema,VegaSchema}, Vega constructs
a dataflow graph based on the specification. It uses the graph to efficiently propagate updates in reaction to user input by traversing the graph.

\subsubsection*{Operators}
The nodes of the graph are called \emph{operators}. Each operator has a value, an optional update function, and a set of \emph{parameters} whose changes trigger a value update. Parameters can be primitive values (which never change) or other operators. 

For example, the \vinl{formula} transformation in an earlier example has parameters \vinl{expr} and \vinl{as}.
If the expression in the \vinl{expr} string refers to a signal name, the graph node (operator) corresponding to the signal becomes a parameter of the \vinl{expr} expression. If the transform is part of a \vinl{transform} definition, then it includes a special parameter named \emph{pulse} that refers to the previous transform in the pipeline (if it exists).

When an operator $v$ depends on an operator $u$ as a parameter, there exists a directed edge from $u$ to $v$. There can be multiple edges between two nodes (e.g., when a single signal name is used in both \vinl{expr} and \vinl{as}), and there can even be self-loops (e.g., a self-incrementing signal that reacts to user clicks). The only important restriction is that, excluding self-loops, the graph must be acyclic, so that there exists a~topological ordering on the operators.

\subsubsection*{Pulses}
At runtime, Vega propagates changes to individual data tuples through the dataflow graph. When the source data changes, it propagates the corresponding change using a \emph{pulse}. Pulses carry references to \emph{changesets}, which in turn contain references to data tuples that were added, removed, or updated. 

Upon receiving a pulse, an operator pulls current values from its parameters and calls its own update function. This updates the operator's value and returns a modified pulse which is then passed to successors. The propagation is orchestrated by a dataflow scheduler which ensures that the changes propagate in topological order (an operator is thus only evaluated once its dependencies have processed their pulses).

\subsubsection*{Mutation and Tracking Copies}
During the processing of pulses, data tuples are announced as being added to, removed from, or modified in the source dataset. In case of modification events, the change is applied directly on the original data tuple, mutating it in-place.

Mutation is only a local measure as the mechanism would easily break if the same source tuples were modified throughout the whole graph. Thus, in a single data pipeline, all data tuples are modified in-place, but on the pipeline's boundaries, copies of original tuples are created, and updates to these copies are redirected using lookup tables. As the tuples can get almost arbitrarily modified throughout the graph, they receive unique numerical identifiers. The primary usage of identifiers is tracking observed tuples in transform operators by using the identifiers as keys in various lookup tables.

\section{The Essence of Vega}
\label{sec:essenc}
We capture the essence of Vega data transformations using a small formal model
that we call the \tinyvega calculus. We discuss possible extensions to the calculus to model
a larger portion of the language (most importantly, signals) in \autoref{subsec:accuracy}.

The syntax of the \tinyvega calculus is shown in \autoref{fig:syntax-new}.
It models the fact that a specification defines new datasets computed from
existing datasets through transformations, includes four basic transformations and
models the restricted expression language used in Vega.

\subsubsection*{Specifications}
A \emph{specification} $s$ consists of a sequence of \emph{dataset} definitions that
define new named datasets by 
transforming earlier datasets through a sequence of
\emph{transforms}~$t$. We refer to the former as input datasets. We assume that those are defined outside of the formal model.

As an example, the following models a subset of the example discussed in \autoref{sec:intro}.
We have an input dataset \textit{books} and we want to count the average
number of words per book \textit{classification} and report the average in hundreds:
\[
\begin{array}{l}
\data(\mathit{books},\mathit{aggs},\\
\quad \agg(\datum.\mathit{classification}, \\
\qquad \avg \datum.\mathit{words} \as \mathit{avg\_words}, \countt\as \mathit{book\_count}), \\
\quad \formula(\round(\mathit{avg\_words}\mathop{/}\gnum{100}), \mathit{avg\_hundreds})~)
\end{array}
\]
The specification assumes an input dataset \textit{books} and defines a new dataset \textit{aggs}.
It performs two transformations. First, it aggregates the data using \textit{classification} as
the key, counting the books and computing \textit{avg\_words} by averaging the number of
\textit{words}. Then, it uses the $\formula$ transformation to define a new field
\textit{avg\_hundreds}, which is computed by dividing the number of words by 100 and rounding the result.

\subsubsection*{Transformations}
Vega uses a streaming-based execution model. When processing a stream of input data,
a transformation takes a single data tuple (datum) and outputs a number of updates to
be propagated to the subsequent steps of the computation. The shape of the output data
tuple depends on the kind of transformation. There are four common kinds:

\begin{enumerate}
  \item those that \textit{keep} the input shape (e.g. filtering)
	\item those that \textit{extend} the input shape
		with additional, newly defined fields (e.g. the formula transform)
	\item those that \textit{nest} the input shape (e.g. cross product)
  \item those that \textit{discard} the input
		shape and produce a different output shape (e.g. aggregation)
\end{enumerate}

\tinyvega includes one example of each kind; $\filter$ emits data for which the
given expression $e$ evaluates to \textsf{true}; $\formula$ evaluates the given expression and
adds the result to the datum as a new field; $\cross$ produces a cross product of the data
stream with itself, storing the two data points in a record as fields $f_a$ and $f_b$.
Finally, $\agg$ groups the input data according to a key consisting of the specified fields
and computes the specified aggregation operations on each group. The resulting data tuples
contain the keys, along with the computed aggregates stored in the output fields $f_{\out}$.

We omit transforms that emit signals and transforms working with GeoJSON and trees. We also omit transforms that combine
multiple data sources (e.g., $\mathsf{lookup}$). Those are infrequent, but are
theoretically interesting as they can introduce \emph{glitches}~(\autoref{subsec:accuracy}).
Finally, we omit transformations where the resulting shape depends on the values
in the input dataset (e.g., $\mathsf{pivot}$) as those cannot be type-checked statically.

\subsubsection*{Expressions}
Vega specifications can use a limited subset of JavaScript to specify computations inside
transforms. Most notably, expressions cannot introduce new variables and can only access the
current data tuple ($\datum$) and its fields, perform numerical operations and call a limited
number of built-in functions.

In \tinyvega, primitives include Booleans, numbers and strings. An access to the data tuple
or its (nested) field is modelled using a path accessor $p$. Simple expressions $a$
can appear as arguments to numerical operators and include values, path accessors, numerical
operators and calls to built-in functions. Expressions $e$ can also construct arrays and records.

\begin{figure}[t]
$\boxed{\textbf{Expression typing:}\quad T \vdash e : T' }$\hfill\mbox{}

\begin{align*}
	T &\vdash b : \bool \quad \text{\sc(T-eBool)}\\
	T &\vdash i : \num \quad \text{\sc(T-eNum)}\\
	T &\vdash s : \str\quad \text{\sc(T-eStr)}\\
  T &\vdash \datum : T\quad \text{\sc(T-eDatum)}
\end{align*}

$$
{ T \vdash p : \{ \ldots, f : T', \ldots \} \over T \vdash p.f : T' }\quad\text{\sc(T-eField)}
$$

$$
\inferrule{ \oplus \in \{+,-,\times,/\}\\\\ T\vdash e_1 : \mathsf{number}\quad T\vdash e_2 : \mathsf{number}}
	 {T \vdash e_1 \oplus e_2: \mathsf{number} }\quad\text{\sc(T-eArithNum)}
$$

$$
\inferrule{ T' \in \{ \str, \num, \bool \}\\\\
    T\vdash e_1 : \str \qquad T\vdash e_2 : T'}
	{T \vdash e_1 + e_2: \str}\quad\text{\sc(T-eAddStrL)}
$$

$$
\inferrule{ T' \in \{ \num, \bool \}\\\\
    T\vdash e_1 : T' \qquad T\vdash e_2 : \str}
	{T \vdash e_1 + e_2: \str}\quad\text{\sc(T-eAddStrR)}
$$

$$
{  \forall i \in 1..n\;.\; T \vdash e_i : T'
\over
 T \vdash [e_1,\ldots,e_n]: [T']
	}\quad\text{\sc(T-eArr)}
$$

$$
{ \forall i \in 1..n\;.\; T \vdash e_i : T_i
\over
\ T \vdash \{f_1\colon e_1,\ldots,f_n\colon e_n\}: \{f_1\colon T_1,\ldots,f_n\colon T_n\}
	}\quad\text{\sc(T-eObj)}
$$

$$
\inferrule{ \Sigma(\mathit{builtin}) = (T_1, \ldots, T_n) \Rightarrow T \\\\
  \forall i \in 1..n\;.\; T \vdash e_i : T_i}
 {T \vdash \mathit{builtin}(e_1, \ldots, e_n): T
	}\quad\text{\sc(T-eBuiltin)}
$$

\vspace{-0.5em}
	\caption{Expression judgment $T \vdash e : T'$.}
\vspace{-1.5em}
	\label{fig:expr}
\end{figure}

\begin{figure}[t]
$\boxed{\textbf{Aggregation operator typing:}\quad T \vdash \op : T' }$\hfill\mbox{}

$${
	T \vdash p_{\inn} : \num
	\over T \vdash \summ p_{\inn} \as f_{\out} : \num
}\quad\text{\sc(T-opSum)}
$$

$$
{\over T \vdash \countt\as f_{\out} : \num
}\quad\text{\sc(T-opCount)}
$$
\\[1.25em]

$\boxed{\textbf{Transform typing:}\quad T \vdash t : T' }$\hfill\mbox{}

$$\mathsf{name}(\datum.f_1.\ldots.f_n) = f_1.\ldots.f_n$$
$$\mathsf{out}(\summ p_{\inn} \as f_{\out}) = f_{\out}$$
$$\mathsf{out}(\countt\as f_{\out}) = f_{\out}$$

$$
{ T \vdash e : T'
\over
T \vdash \filter(e) : T
	}\quad\text{\sc(T-tFilter)}
$$

$$
{ T \vdash e : T_{\out}
\over
T \vdash \formula(e,f_{\out}) : T \cup \{ f_{\out} : T_{\out}\}
}\quad\text{\sc(T-tFormula)}
$$

$$
{
\over
T \vdash \cross(f_a, f_b) : \{f_a : T, f_b : T\}
	}\quad\text{\sc(T-tCross)}
$$

$$
\inferrule
 {\forall i \in 1..n\;.\;T \vdash p_i : T_i \quad \forall j \in 1..k\;.\;T \vdash \op_j : T'_j \\\\
  T_r = \{ \mathsf{name}(p_1) : T_1, \ldots,\mathsf{name}(p_n) : T_n,\\
    \phantom{xxxx}\mathsf{out}(\op_1) : T'_1, \ldots, \mathsf{out}(\op_k) : T'_k
    \}
  }
 {T \vdash \agg( p_1,\ldots,p_n,\op_1,\ldots,\op_k) : T_r }\quad\text{\sc(T-tAgg)}
$$\\[1.25em]


$\boxed{\textbf{Dataset definition typing:}\quad \Gamma \vdash d : T }$\hfill\mbox{}

$$
{ \forall i\in 1..k\;.\; T_{i-1} \vdash t_i: T_i \qquad \Gamma(n_{\rm in}) = T_0
\over
\Gamma \vdash \data(n_{\out}, n_{\inn}, t_1,\ldots,t_k): T_k
	}\quad\text{\sc(T-dData)}
$$
\\[1.25em]

$\boxed{\textbf{Specification typing:}\quad \Gamma \vdash \spec : \Gamma' }$\hfill\mbox{}

$$
\inferrule{
    \forall i\in 1..n\;.\; \Gamma_{i-1} \vdash d_i : T_i \\\\
     d_i = \data(\_, n_i, \ldots)\qquad \Gamma_i =\Gamma_{i-1}, n_{i}\!:\!T_i}
  {\Gamma_0 \vdash d_1,\ldots,d_n : \Gamma_n
	}\quad\text{\sc(T-spec)}
$$

\vspace{-0.5em}
\caption{Judgments that define the \tinyvega type system.}
\label{fig:types}
\end{figure}

\section{Type System}
\label{sec:types}

A well-typed \tinyvega specification does not access undefined fields or datasets not defined
previously. Our type system captures the idea formally, focusing on how data transformations
change the shape of data tuples. Our system recognizes a number of primitive types, records and
collections:
\begin{align*}
	\prim &\de \bool \orr \num \orr \str \\
	T &\de \prim \orr \{f_1: T_1, \ldots, f_n: T_n\} \orr [T]
\end{align*}
We define the type system using a series of judgments for individual syntactic categories of
\tinyvega; \autoref{sec:sem} adds a graph-based operational
semantics and \autoref{sec:sound} discusses soundness.

\subsubsection*{Expression judgment} The judgment $T\vdash e : T'$, defined in \autoref{fig:expr},
states that, upon receiving a $\datum$ of type $T$, the expression $e$ outputs a new $\datum$ of type $T'$. In the standard notation, this translates to $\{\datum : T\} \vdash e : \{ \datum : T'\}$.

Vega expressions can access only the $\datum$ variable (and signals and builtins, which we omit).
The {\sc T-eDatum} and {\sc T-eField} rules model access to the $\datum$ object and its fields.
We match the semantics of JavaScript by allowing the use of $+$ for string concatenation with an
implicit conversion for numbers and booleans ({\sc T-eAddStrL} and {\sc T-eAddStrR}).
Vega expressions can also create new collections ({\sc T-eArr}) and objects ({\sc T-eObj})
and call a range of primitive functions ({\sc T-eBuiltin}). We assume their types are given in $\Sigma$.

\subsubsection*{Transform judgment}
The judgment $T \vdash t : T'$ (\autoref{fig:types}) states that, given a data tuple of shape $T$,
the transform $t$ produces a data tuple of shape $T'$.
The formula transform ({\sc T-tFormula}) computes a value of type $T_{\out}$ and adds it as a~new field to the data tuple. We use the $\cup$ operator to denote the join of two data tuples on their keys, with the right-hand-side keys taking precedence.
Cross product ({\sc T-tCross}) creates a~new data tuple that contains two values of the
input type.
The output type of the filter transform ({\sc T-tFilter}) is the same as the input type,
but we keep the JavaScript semantics that allows the conditional to evaluate to a value of any type.

Finally, aggregation ({\sc T-tAgg}) produces a tuple that combines the fields used as grouping
keys with new fields defined by the aggregation operations ({\sc T-opSum} and {\sc T-opCount}).
The names of the keys are generated using the $\mathsf{name}$ helper, by concatenating the names
of the nested fields accessed. (In Vega, a field such as ``a.b'' is accessed using
\texttt{datum["a.b"]} to avoid an ambiguity with nested field access.) The names of new fields
are those specified in the aggregation operator, obtained by the $\mathsf{out}$ helper.

\subsubsection*{Dataset judgment}
For dataset and specification judgments, we use a typing context $\Gamma$, which maps
the dataset names to the type of the data tuples included in the dataset. (The type is the
type of the individual elements in the dataset.)

The rule {\sc T-dData} that defines the dataset judgment is shown in \autoref{fig:types}.
It specifies the type of data tuples for a~newly defined output dataset named $n_{\out}$
that is computed by a series of transformations from the input dataset named $n_{\in}$.
The sequence of types $T_0, \ldots, T_k$ starts from the type of the source and gradually
transforms the type using the transform judgment.

\subsubsection*{Specification judgment}
Finally, the judgment $\Gamma \vdash s : \Gamma'$ defined in \autoref{fig:types}
states that a specification $s$ transforms the initial dataset context $\Gamma_0$ to a newly
defined dataset $\Gamma_n$. This is done by iteratively using the dataset judgment and
adding the newly defined dataset to the previous typing context $\Gamma_{i-1}$.
Note that the definition models the fact that dataset definitions are sequential
and that datasets might be referenced only after their declaration.

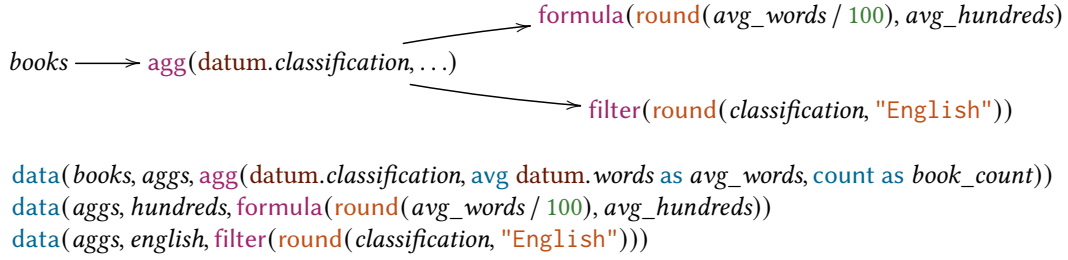
\begin{figure*}
\centering
$\xymatrix@R=0.2pc{
  && \formula(\round(\mathit{avg\_words}\mathop{/}\gnum{100}), \mathit{avg\_hundreds}) \\
\mathit{books}\ar[r] &
  \agg(\datum.\mathit{classification},\ldots) \ar@/^0.5pc/[ur] \ar@/_0.5pc/[dr] &
  \\
  && \filter(\contains(\mathit{classification}, \text{\color{Bittersweet} \texttt{"English"}}))
}$
\\[1em]
$
\begin{array}{l}
\data(\mathit{books},\mathit{aggs}, \agg(\datum.\mathit{classification},\avg \datum.\mathit{words} \as \mathit{avg\_words}, \countt\as \mathit{book\_count}))\\
\data(\mathit{aggs}, \mathit{hundreds}, \formula(\round(\mathit{avg\_words}\mathop{/}\gnum{100}), \mathit{avg\_hundreds}))\\
\data(\mathit{aggs},\mathit{english}, \filter(\contains(\mathit{classification}, \text{\color{Bittersweet} \texttt{"English"}}))) \\
\end{array}
$
\vspace{-0.25em}
\caption{A \tinyvega specification that defines three data transformations and a corresponding
  graph. Vertices are annotated with the transformations they perform and edges model data flow
  in the graph.}
\label{fig:graph}
\end{figure*}

\section{Operational Semantics}
\label{sec:sem}

Our model of the Vega runtime is mainly based on the high-level description of Reactive
Vega \cite{reactive-vega-architecture-2016}, but it reflects some aspects of the implementation to show how real Vega transforms can be formalized using our model. The overall idea is based on propagating data tuple changes in a dataflow graph. The graph construction is akin to other
reactive and live programming environments \cite{Petricek2020,oeyen-2023-graph}, but it is
then used to propagate changes (adding, removing or updating data tuples).

The graph structure is illustrated in \autoref{fig:graph}. Vertices in the graph
correspond to data transformations and the edges represent how data flows between
them.

As our ultimate goal is to statically check dataset field access, we restrict ourselves to dataflow graphs with only transformation operators (i.e. 'pulse' parameters). We discuss how our model can be extended to allow limited usage of user input in~\autoref{subsec:accuracy}.

Formally, a \tinyvega dataflow graph $(V,E)$ consists of:
\[
\begin{array}{r@{\,}c@{\,}ll}
	V &\densp& \{(l_1,t_1)\ldots\}&\textit{transform $t_i$ with a unique label $l_i$}\\
	E &\subseteq& \{(v_i,v_j), i,j \in V \}& \textit{data from $v_i$ flows to $v_j$}\\
\end{array}
\]
To distinguish vertices that perform the same transformation in different parts
of the dataflow graph, we annotate them with a unique label $l_i$.

\subsection{Dataflow graph construction}
The process of dataflow graph construction for a given \tinyvega specification is
defined in \autoref{fig:builders}. It is structured in terms of three functions
for translating a transformation, dataset and a specification, respectively.
Each of these functions takes a partially built graph and produces an updated graph
with newly added vertices and edges.

\subsubsection*{Transform builder}
The $\buildT$ function is called with the graph constructed so far, a vertex $v$ representing
the data source and a transform. As illustrated in \autoref{fig:graph}, it constructs a new
vertex annotated with the transformation and a fresh label and an edge connecting it to the
source. The function then returns the modified graph and the new vertex.

\subsubsection*{Dataset builder}
$\buildD$ constructs a new graph in steps
$(V_1, E_1), \ldots, (V_k, E_k)$, by adding a vertex for each of the individual
transformations using $\buildT$. It operates on a $\lookup$ function
that maps dataset names to corresponding vertices, using it to find the vertex corresponding
to the input dataset~$n_{\inn}$ and adding a mapping for the output vertex $n_{\out}$ at the end.
Note that $\lookup(n_{\inn})$ is not defined if the input dataset does not appear earlier
in the Vega specification.

\subsubsection*{Specification builder}
Finally, $\buildS$ takes a sequence of dataset definitions. It constructs the graph iteratively
using $\buildD$ to process individual dataset definitions, constructing
a new graph and a lookup function that maps names of all datasets to the
corresponding vertex in the graph.

\vspace{0.5em}
\noindent
\tinyvega does not model primitive data sources such as files or URLs. We assume that the
specification uses
a set of pre-existent sources named ${n_{\inn} \in \mathcal{S}}$ and we construct special vertices
representing the input nodes annotated with $(n_{\inn}, \inputt)$. An initial graph is constructed by the function $\mathsf{build}(\mathcal{S}, \spec)$ as:
{
\renewcommand{\arraystretch}{1.2}
\[
\begin{array}{l}
\mathsf{build}(\mathcal{S}, \spec) = (V, E), \lookup\;\;\mathit{where}\\
\quad \lookup_0 = \{(n_{\inn} \mapsto (n_{\inn}, \inputt))\ |\ n_{\inn} \in \mathcal{S} \}\\
\quad V_0, E_0 = \{ (n_{\inn}, \inputt)\ |\ n_{\inn} \in \mathcal{S} \}, \emptyset\\
\quad (V, E), \lookup = \buildS_{(V_0, E_0), \lookup_0}(\spec)
\end{array}
\]
}

\begin{figure}
\raggedright
{
\renewcommand{\arraystretch}{1.2}
$\boxed{\textbf{Transform builder}\;\; \buildT}$
\\[1em]

$\begin{array}{l}
\buildT : (\mathbb V \times \mathbb E, V, t) \to (\mathbb V\times \mathbb E, V) \\[0.3em]
\buildT_{(V,E), v}(t) \coloneq (V \cup \{ v' \}, E' \cup \{ (v, v') \}), v'\\
\quad \textit{where}~v'=(l, t) ~\textit{and}~l~\textit{is a fresh label}
\end{array}$

\; \\[1em]
$\boxed{\textbf{Dataset builder}\;\; \buildD}$
\; \\[1em]

$\begin{array}{l}
\buildD : (\mathbb V \times \mathbb E, N \to V, d) \to (\mathbb V\times \mathbb E, N \to V)\\[0.3em]
\buildD_{(V_0,E_0), \lookup} (d) \coloneq (V_{k}, E_{k}), \lookup', v_{k} \textit{ where}\\
\quad \data(n_{\out}, n_{\inn}, t_1, \ldots, t_k) = d\\
\quad v_0 = \lookup(n_{\inn}),\\
\quad (V_i, E_i), v_i = \buildT_{(V_{i-1},E_{i-1}), v_{i-1}}(t_i)\quad (\forall i \in 1..k)\\
\quad \lookup' = \lookup \cup \{(n_{\out}, v_{k})\}
\end{array}$

\; \\[1em]
$\boxed{\textbf{Specification builder}\;\; \buildS}$
\; \\[1em]

$\begin{array}{l}
\buildS : (\mathbb V \times \mathbb E, N \to V, s) \to (\mathbb V\times \mathbb E, N \to V)\\[0.3em]
\buildS_{(V_0,E_0),\lookup_0}(d_1,\ldots,d_k) \coloneq (V_k, E_k), \lookup_k \textit{ where}\\
\quad (V_i, E_i), \lookup_i = \buildD_{(V_{i-1},E_{i-1}),\lookup_{i-1}}(d_i)\quad (\forall i \in 1..k)
\end{array}$
}
\caption{Builders that define data-flow graph construction for transforms,
datasets and \tinyvega specification.}
\vspace{-1em}
\label{fig:builders}
\end{figure}

\subsection{Graph-based evaluation}

The key aspect of the Reactive Vega implementation \cite{reactive-vega-architecture-2016} that
\tinyvega aims to model accurately is the graph-based streaming propagation of updates.
When a data tuple is added, removed, or modified at the data source, the execution engine
propagates the update through the dataflow graph and recomputes values of all datasets
that are derived from the source through a series of transformations.

\subsubsection*{Data tuple identity}
In the implementation of Vega \cite{Vega-v530}, data tuples are mutable objects with numerical identifiers stored under a special hidden key (symbol in JavaScript terminology). The implementation uses identifiers to propagate data changes to corresponding copies (e.g., the Relay transform) and to new objects containing the original tuples (e.g., TreeLinks transform).
Moreover, identifiers are used to maintain various lookup tables in stateful transformations.

Our model also needs to uniquely identify data tuples. Although the definitions do not directly leverage identifiers in the same way as Vega, they are important in proving the soundness of our semantics. Furthermore, the existence of identifiers makes for semantics that can be extended to model more Vega transformations.


We address this by annotating each data tuple with a symbolic trace $R$, inspired by
traces in provenance tracking \cite{perera-2012-explain}. As shown in \autoref{fig:eval-objects}, a trace $R$ can be a primitive number $n$, or a
term that uniquely represents a transformation applied to a data tuple; $(R_1, R_2)$ identifies a data tuple produced by joining two data tuples using $\cross$,
while $f_e(R)$ identifies a data tuple produced by applying $\formula$ which introduces the field $f$ computed by the expression $e$. Finally, $\aggb_{\mathbf{p}, c, \mathbf{sum}}(R)$ denotes a group produced by aggregating tuples using the paths $\mathbf{p}, c, \mathbf{sum}$ as described in \autoref{subsec:transeval}.


\subsubsection*{Values and changes}
\autoref{fig:eval-objects} defines values, data tuples and data changes. A value $v$ represents
primitive values and values computed by Vega expressions. 
A data tuple is an object (record) with a~trace field $\mu\!:\!R$. Choosing a special name $\mu$ corresponds to the fact that Vega stores identifiers under a symbol that cannot be accessed by a user.
As such, traces are not present in our type system, but we assert their properties in \autoref{sec:sound}
. Note that traces are only
needed at the top-level. There are three data changes (add, update, remove), and we also include
the $\error$ change that is triggered when the evaluation fails due to invalid field access.

\subsubsection*{Expression evaluation}
We assume a standard big-step evaluation relation $\Downarrow$ for Vega expressions.
Recall that expressions do not contain variables, but may contain a reference to $\datum$.
We write $e[\datum/v]$ for substitution that replaces the $\datum$ with a value $v$.
We write $e\!\Downarrow\!v$ if an expression $e$ evaluates to $v$ and
$e\cancel{\Downarrow}$ if the expression cannot be evaluated, e.g., because it attempts to
access a~non-existent record field.

\begin{figure}
  $\boxed{\textbf{Data tuples:}~$d$\,}$\hfill\mbox{}
  \\[-0.5em]
  \begin{alignat*}{3}
  		\text{Traces}\quad&&R&~\de n \orr (R_1, R_2) \orr f_e(R) \orr \mathsf{agg}_{\mathbf{p}, c, \mathbf{sum}}(R) \\
  		\text{Values}\quad&&\mathsf{v}&~\de b \in \texttt{Boolean} \orr n \in \texttt{Number} \orr s \in \texttt{String} \\
                        && &\;\;\;\orr\;\;\; [v_1,\ldots,v_n] \orr\{f_1\colon v_1, \ldots, f_n\colon v_n\}\\
      \text{Data tuples}\quad&&\mathsf{d}&~\de \{f_1\colon v_1, \ldots, f_n\colon v_n,\mu:R\}\\
  \end{alignat*}
  \\[-0.5em]

  $\boxed{\textbf{Data changes:}~$u$\,}$\hfill\mbox{}
  \\[-0.5em]
  \begin{align*}
  	\add(d)&\quad\text{data tuple $d$ has been created}\\
  	\update(d,d')&\quad\text{replace data tuple $d$ with $d'$; keep $\mu$}\\
  	\remove(d)&\quad\text{data tuple $d$ has been removed}\\
  	\error&\quad\text{error caused by invalid field access} 
  \end{align*}

\caption{Definition of runtime data tuples and data updates.}
\label{fig:eval-objects}
\end{figure}

\begin{figure*}
{
\renewcommand{\arraystretch}{1.1}
\renewcommand{\truthy}{\mathsf{truthy}}
\renewcommand{\falsy}{\mathsf{falsy}}
$\boxed{\textbf{Formula transformation}}$\hfill\mbox{}
\\[-0.75em]
\[
\begin{array}{r@{\;\;}c@{\;\;}l@{\quad}l}
v +_{f, e} v' &=& v \cup \{ f = v', \mu = f_e(v.\mu) \}\\[0.5em]
\llbracket \formula_{e, f}\rrbracket(\add(v)) &=& \{ \add(v +_{f,e} v_e )\} &(e[\datum/v] \Downarrow v_e)\\
\llbracket \formula_{e, f}\rrbracket(\remove(v)) &=& \{ \remove(v +_{f,e} v_e) \} &(e[\datum/v] \Downarrow v_e)\\
\llbracket \formula_{e, f}\rrbracket(\update(v, v_\new)) &=& \{ \update(v +_{f,e} v_e, v_\new +_{f,e} v_e') \} &(e[\datum/v] \Downarrow v_e \wedge e[\datum/v_\new]\ \Downarrow v_e')\\
\llbracket \formula_{e, f}\rrbracket(\_) &=& \{ \error \} &(\text{otherwise})\\
\end{array}
\]
\\[0.5em]

$\boxed{\textbf{Cross transformation}}$\hfill\mbox{}
\\[-0.75em]
\renewcommand{\state}{\mathsf{st}}
\[
\begin{array}{r@{\;\;}c@{\;\;}l}
v_a \times_{f_a, f_b} v_b &=& \{ f_a = v_a, f_b = v_b, \mu = (v_a.\mu, v_b.\mu) \}\\[0.5em]
\llbracket \cross_{f_a,f_b} \rrbracket_{\state}(\add(v)) &=&
      \{ \add(v \times_{f_a, f_b} v'), \forall v' \in \state\cup\{v\} \} \cup
      \{ \add(v'\times_{f_a, f_b} v), \forall v' \in \state \}, \;
      \state \cup \{v\}
\\
\llbracket \cross_{f_a,f_b} \rrbracket_{\state}(\remove(v)) &=&
      \{ \remove(v \times_{f_a, f_b} v'), \forall v' \in \state \} \cup
      \{ \remove(v' \times_{f_a, f_b} v), \forall v' \in \state\setminus\{v\}  \}, \;
      \state \setminus \{v\}
\\
\llbracket \cross_{f_a,f_b} \rrbracket_{\state}(\update(v, v_\new)) &=&
      \{ \update(v \times_{f_a, f_b} v', v_\new \times_{f_a, f_b} v'), \forall v' \in \state \} ~\cup \\&&
      \quad \{ \update(v' \times_{f_a, f_b} v, v' \times_{f_a, f_b} v_\new), \forall v' \in \state\setminus \{ v \} \},
      (\state \setminus \{ v \}) \cup \{ v_\new \}\\
\llbracket \cross_{f_a,f_b} \rrbracket_{\state}(\error) &=& \{ \error \}, \state \\
\end{array}
\]
\\[0.5em]

$\boxed{\textbf{Filter transformation}}$\hfill\mbox{}
\renewcommand{\state}{\mathsf{state}}
\\[-0.5em]
\[
\begin{array}{r@{\;\;}c@{\;\;}l@{\quad}l}
  \llbracket \filter_e \rrbracket(\add(v)) &=& \{ \add(v) \} &(e[\datum/v] \Downarrow v_e \wedge \truthy(v_e))\\
  \llbracket \filter_e \rrbracket(\add(v)) &=& \{ \} &(e[\datum/v] \Downarrow)\\[0.5em]
  \llbracket \filter_e \rrbracket(\remove(v)) &=& \{ \remove(v) \} &(e[\datum/v] \Downarrow v_e \wedge \truthy(v_e))\\
  \llbracket \filter_e \rrbracket(\remove(v)) &=& \{ \} &(e[\datum/v] \Downarrow)\\[0.5em]
  \llbracket \filter_e \rrbracket(\update(v, v_\new)) &=& \{ \update(v, v_\new) \} &(e[\datum/v] \Downarrow v_e \wedge \truthy(v_e) \wedge e[\datum/v_\new]\ \Downarrow v_e' \wedge \truthy(v_e'))\\
  \llbracket \filter_e \rrbracket(\update(v, v_\new)) &=& \{ \add(v_\new) \} &(e[\datum/v] \Downarrow v_e \wedge \falsy(v_e) \wedge e[\datum/v_\new]\ \Downarrow v_e' \wedge \truthy(v_e'))\\
  \llbracket \filter_e \rrbracket(\update(v, v_\new)) &=& \{ \remove(v) \} &(e[\datum/v] \Downarrow v_e \wedge \truthy(v_e) \wedge e[\datum/v_\new]\ \Downarrow v_e' \wedge \falsy(v_e'))\\
  \llbracket \filter_e \rrbracket(\update(v, v_\new)) &=& \{ \} &(e[\datum/v] \Downarrow \wedge\ e[\datum/v_\new] \Downarrow)\\[0.5em]
  \llbracket \filter_e \rrbracket(\_) &=& \{ \error \} &(\text{otherwise})\\
\end{array}
\]\\[0.5em]

\newcommand{\newGroup}{\mathsf{newGroup}}

$\boxed{\textbf{Aggregate transformation}}$\hfill\mbox{}
\renewcommand{\state}{\mathsf{st}}
\renewcommand{\counta}{\mathsf{c}}
\\[-0.5em]
\[
\begin{array}{r@{\;\;}c@{\;\;}l@{\quad}l}
  v +_{\counta, \mathbf p} v' &=& \{\counta = v.\counta + 1, {\mathbf p}_{\out} = v.{\mathbf p}_{\out} + v'.{\mathbf p}_{\inn}, \mu = v.\mu \} \\
  v -_{\counta, \mathbf p} v' &=& \{\counta = v.\counta - 1, {\mathbf p}_{\out} = v.{\mathbf p}_{\out} - v'.{\mathbf p}_{\inn}, \mu = v.\mu \} \\
  \phantom{\ \agg}\newGroup_{\mathsf{groupby}, \mathsf{count}, \mathbf p}(v) &=& \{\counta = 1, {\mathsf f}_{\out} = v.{\mathsf f}_{\inn}, \mu = \mathsf{agg}_{\mathsf{groupby}, \mathsf{count}, \mathsf f}(v.\mu) \}\\
  \groupOf_{\mathbf p}(\state, v) &=& \state(v.{\mathbf p})\end{array}\]\\[-0.5em]
  \[
  \begin{array}{r@{\;\;}c@{\;\;}l@{\quad}l}
  \llbracket \agg_{\mathbf p, \counta, \sumi}\rrbracket_{\state}(\add(v)) &=& \{\add(g)\}, \state + (v.{\mathsf f}, g) & (v.\mathbf p \not\in \state, g=\newGroup_{\mathbf {f}, \counta, \sumi}(v))\\
  \llbracket \agg_{\mathbf p, \counta, \sumi}\rrbracket_{\state}(\add(v)) &=& \{\update(g, g')\}, \state + (v.{\mathbf p}, g') & (v.\mathbf p \in \state, g=\state(v.\mathbf p), g' = g +_{\counta, \sumi} v)\\[0.5em]
\llbracket \agg_{\mathbf p, \counta, \sumi}\rrbracket_{\state}(\remove(v)) &=& \{\update(g, g')\}, \state + (v.{\mathsf f}, g') & (g=\state(v.\mathbf p), g.\counta > 1, g' = g -_{\counta, \sumi} v)\\
\llbracket \agg_{\mathbf p, \counta, \sumi}\rrbracket_{\state}(\remove(v)) &=& \{\remove(g)\}, \state - (v.{\mathsf f}, g) & (g=\state(v.\mathbf p), g.\counta=1)\\[0.5em]
\llbracket \agg_{\mathbf p, \counta, \sumi}\rrbracket_{\state}(\update(v, v_\new)) &=& \{\update(g, g')\}, \state + (v.{\mathsf f}, g') & (v.\mathbf p = v_\new.\mathbf p, g=\state(v.\mathbf p),
g' = g -_{\counta, \sumi} v +_{\counta, \sumi} v_\new)
\\
\llbracket \agg_{\mathbf p, \counta, \sumi}\rrbracket_{\state}(\update(v, v_\new)) &=& \parbox[t]{4cm}{$\{\remove(g), \add(g')\}$,\\
\phantom{\{r}$\state - (v.{\mathsf f}, g) + (v_\new.{\mathsf f}, g')$}
& \parbox[t]{6cm}{$(v.\mathbf p \ne v_\new.\mathbf p, g=\state(v.\mathbf p), g.\counta = 1$,\\
\phantom{(v}$v_\new.{\mathbf p} \not\in\state, g'=\newGroup_{\mathbf p, \counta, \sumi}(v_\new) )$}
\\
\llbracket \agg_{\mathbf p, \counta, \sumi}\rrbracket_{\state}(\update(v, v_\new)) &=& \parbox[t]{4cm}{$\{\remove(g), \update(g', g'')\},$\\
\phantom{\{r}$\state - (v.{\mathsf f}, g) + (v_\new.{\mathsf f}, g')$}
& \parbox[t]{6.5cm}{$(v.\mathbf p \ne v_\new.\mathbf p, g=\state(v.\mathbf p), g.\counta = 1, v_\new.{\mathbf p} \in\state$,\\ 
\phantom{(v}$g'=\state(v_\new), g'' =g'+_{\counta, \sumi}v_\new)$}
\\
\llbracket \agg_{\mathbf p, \counta, \sumi}\rrbracket_{\state}(\update(v, v_\new)) &=& \parbox[t]{4cm}{$\{\update(g, g''), \add(g')\},$\\ 
\phantom{\{u}$\state + (v.{\mathsf f}, g'') + (v_\new.{\mathsf f}, g')$}
& \parbox[t]{6.5cm}{$(v.\mathbf p \ne v_\new.\mathbf p, g=\state(v.\mathbf p), g.\counta > 1, v_\new.{\mathbf p} \not\in\state$,\\
\phantom{(v}$g'=\newGroup_{\mathbf p, \counta, \sumi}(v_\new), g'' =g-_{\counta, \sumi}v)$}
\\
\llbracket \agg_{\mathbf p, \counta, \sumi}\rrbracket_{\state}(\update(v, v_\new)) &=& \parbox[t]{4.5cm}{$\{\update(g, g''), \update(g', g''')\},$\\
\phantom{\{u}$\state + (v.{\mathsf f}, g'') + (v_\new.{\mathsf f}, g''')$}
& \parbox[t]{7cm}{$(v.\mathbf p \ne v_\new.\mathbf p, g=\state(v.\mathbf p), g.\counta > 1, v_\new.{\mathbf p} \in\state,$\\
\phantom{(v}$g'=\state(v_\new), g'' =g-_{\counta, \sumi}v, g''' =g+_{\counta, \sumi}v_\new)$}
\\
\llbracket \agg_{\mathbf p, \counta, \sumi}\rrbracket_{\state}(\_) &=& \error, \state &  \text{(otherwise)}

\end{array}
\]

\caption{Semantics of update propagation of core \tinyvega transformations}
\label{fig:core-transforms}
}
\end{figure*}

\subsection{Transform evaluation}
\label{subsec:transeval}
When a graph node receives a data change, it produces 0, 1 or more data changes that are sent
to subsequent nodes. We first define the evaluation of each of our four transformations.
There are two kinds of transformations. Stateful transformations need to maintain an internal
state, associated with the graph node. Stateless transformations process data changes without
keeping a state.
The definitions of core transformations are given in Figure~\ref{fig:core-transforms} and are
written as:
\[
\begin{array}{r@{\;\;=\;\;}ll}
\llbracket t \rrbracket_{\mathsf{state}}(u) & \{ u_1, \ldots, u_k \}, \mathsf{state}' & (\text{stateful transform})\\
\llbracket t \rrbracket(u) & \{ u_1, \ldots, u_k \} & (\text{stateless transform})
\end{array}
\]
For a given transform $t$, the semantic function receives a data update $u$ (and optionally
a transform-specific $\mathsf{state}$) and produces a set of downstream data
changes $u_1,\ldots,u_k$, possibly with an updated $\mathsf{state}'$. Furthermore, every transformation propagates an $\error$ upon receiving it.
Vertices labeled with $\inputt$ act as data sources and never receive any updates.

\subsubsection*{Formula}
The $\formula$ transform (\autoref{fig:core-transforms}) extends a data tuple $v$ with a
field $f$ obtained by evaluating the expression $e$. An important aspect of the semantics
is maintaining data tuple identity represented by a trace $v.\mu$. The helper operator $+_{f_e}$,
which adds the field $f$ also adds a new trace obtained from the original trace by wrapping it
with $f_e(v.\mu)$. As a result, if the transformation receives multiple updates to a data tuple
with the same identity, the downstream updates will also preserve the data tuple identity.

The $\formula$ rule evaluates the formula using $\Downarrow$. In case of an $\update(v,v_\new)$,
we reevaluate the previous value to produce a new downstream $\update$. If the evaluation
fails, the semantics produces an $\error$.

\subsubsection*{Cross}
The $\cross$ transform (\autoref{fig:core-transforms}) needs to maintain a~cross product of all
received data tuples with themselves, i.e., a~set $\{ (d_1, d_2), d_1\!\in\!D, d_2\!\in\! D \}$.
It uses the $\state$ to keep track of all values currently in the upstream dataset.
When receiving an update, the transform outputs the new state along with the set of updates.

To construct downstream updates with stable identity, the semantics defines a helper $\times_{f_a, f_b}$.
When it receives $\add(v)$, it computes cross product of $v$ with every other element $v' \in \state$,
yielding both $v \times v'$ and $v' \times v$ and also adding $v \times v$.
Handling of remove and update works similarly. 

\subsubsection*{Filter}
The $\filter$ transform (\autoref{fig:core-transforms}) produces a downstream dataset containing
only upstream data tuples for which the specified condition holds. We assume a pair of predicates
that model the JavaScript semantics of conditionals ($\mathsf{truthy}$ holds for non-zero numbers,
non-empty strings and all object and array values).

Changes $\add$ and $\remove$ are sent downstream if the condition holds. The handling of $\update$ needs to determine if the data tuple has been newly added (condition did
not hold previously, but holds now), updated or removed. Note that the transform always emits a
set of updates (possibly empty) if the evaluation succeeds. If the evaluation fails (or the
upstream update is $\error$, just like in other transformations), the transformation produces the $\error$ change.

\subsubsection*{Aggregate}
The $\agg$ transform (\autoref{fig:core-transforms}) groups the input tuples by the values at paths $p_1,..,p_n$.
For each of the groups, it computes its sum or count (we omit average, as it can be calculated by combining these two values) and outputs them in fields $\out(\op_1), \ldots, \out(\op_k)$. As it only makes sense to count groups at the root path, we assume there is at most one $\countt$ operation present in $\op_1, \ldots, \op_k$. To save space, we introduce boldface for vector notation:
\begin{itemize}
\item[--] $\mathbf p$ denotes $p_1,\ldots,p_n$
\item[--] $\mathbf a + \mathbf b$ denotes $a_1 + b_1, \ldots, a_n + b_n$ (assuming vectors $\mathbf a, \mathbf b$ of the same length $n$)
\item[--] $v.\mathbf f$ denotes either $\{v.\mathbf f_1, \ldots v.\mathbf f_n\}$ or its expansion $v.\mathbf f_1,\linebreak, \ldots, v.\mathbf f_n$ if already present in another object
\item[--] $v.\mathbf f = \mathbf f'$ then means $v.f_1 = f'_1, v.f_2 = f'_2, \ldots, v.f_n = f'_n$
\item[--] Finally, assuming $\mathbf f = \op_1,\ldots,\op_k$, $\mathbf f_{\out}$ denotes $\out(f_1),\linebreak \out(f_2), \ldots, \out(f_n)$ (and similarly for $\mathbf f_{\inn}$).
\end{itemize}
Putting this all together, we define the semantics of $\agg$ using these three parameters:
\begin{itemize}
\item[--] $\mathbf p$ (the paths to group the tuples by), 
\item[--] $\counta$ or $c$ (name of the field representing the group size),
\item[--] and $\sumi$ (the paths at which to compute a sum)
\end{itemize}
The $\state$ is a mapping from possible values at $\mathbf p$ to the currently produced nonempty groups. The group of a tuple $v$ is then accessed via $\state(v.\mathbf p)$, and its existence is tested by $v.\mathbf p \in \state$.
Note that both of these actions result in $\error$ if some of the paths from $\mathbf p$ do not exist in $v$.

When $\agg$ receives $\add(v)$, it creates a new group consisting only of the tuple itself (using the identity $v$ for the new group), or it updates the statistics for an existing group. Similarly in $\remove(v)$, we distinguish whether we should just $\update$ the target group or propagate a $\remove$ change. The cases in $\update$ are similar to those in $\filter$. Either the group of the tuple has not changed, and so we just update the $\sumi$s accordingly. Otherwise, the group has changed, so we need to remove from one group and add to another. We consider 4 cases, handling the situations where the old group becomes empty or the new group does not exist yet.

This whole transformation has one subtle source of errors.
If we somehow $\remove$ or $\update$ a nonexistent tuple~$v$, none of the cases for $\remove$ or $\update$ respectively will match, and the evaluation instead falls-through to the error-propagation state. We address this issue in the graph evaluation algorithm, but it is important to realize that this is why tuple identities are useful in our model.












\subsection{Graph evaluation}
\label{sec:semantics-grapheval}
We define change propagation in a manner close to Vega, which means we propagate the changes with respect to a~topological ordering $\prec$ on vertices. The only additional information we track in our model are the observed tuples in each node, which helps us prove safety.
%
%
%
Let us proceed by defining the state and update logic of the graph. 

Let $\gs : V \to \DataChange[]$ represent the list of unprocessed data changes
in each vertex,
and $\ls : V \to S$ represent the local state for each vertex (if needed),
maintaining the state of stateful transformations. Furthermore, let $\invals : V \to R \to d$ represent the data tuples currently residing in the graph nodes.
We define a step of reduction $\langle \gs, \ls, \invals\rangle \tto \langle \gs', \ls',\invals'\rangle$ 
where we process a~single event in our graph and transform the current state into a~new one.

\begin{enumerate}
	\item Select and remove the first data change $c$ from $\gs(v)$ where $v$ is the smallest w.r.t. $\prec$
    \begin{itemize}
    \item[$\bullet$] If there are no more data changes, return the original $\gs$ and $\ls$.
    \end{itemize}
    \item If $c = \add(t)$:
    \begin{itemize}
    \item[--] If $t.\mu \in \invals(v)$ then $\langle \gs, \ls, \invals\rangle \not{\tto}$
    \item[--] otherwise: $\invals' \gets \invals + (v, \invals(v) + (t.\mu, t))$
    \end{itemize}
    \item If $c = \update(t, t')$:
    \begin{itemize}
    \item[--] If $t.\mu \not\in \invals(v)$ or $\invals(v)(t.\mu) \ne t$:
    \begin{itemize}\item[$\bullet$] $\langle \gs, \ls, \invals\rangle \not{\tto}$\end{itemize}
    \item[--] otherwise: $\invals' \gets \invals + (v, \invals(v) + (t.\mu, t'))$
    \end{itemize}
    \item If $c = \remove(t)$:
    \begin{itemize}
    \item[--] If $t.\mu \not\in \invals(v)$ or $\invals(v)(t.\mu) \ne t$:\begin{itemize}\item[$\bullet$]$\langle \gs, \ls, \invals\rangle \not{\tto}$\end{itemize}
    \item[--] otherwise $\invals' \gets \invals + (v, \invals(v) - (t.\mu, t))$
    \end{itemize}
	\item Let $\llbracket t, \ls(v), c \rrbracket = C, \text{newState}$
	\item $\ls' \gets \ls + (v, \text{newState})$
    \item $\gs' \gets \gs$
    \item \text{For each $v'\in V$ such that there exists an edge $(v, v')$:}
    \begin{itemize}
        \item[--] $\gs'(v') \gets \gs'(v') \cup C$
    \end{itemize}
\end{enumerate}
 
We say $\langle \gs_0, \ls_0, \invals_0 \rangle \tto^n \langle \gs_n, \ls_n,\invals_n\rangle$
if there exists a sequence $\gs_i, \ls_i, \invals_i$
such that $\forall i \in 0..n$, it holds that $\langle \gs_i,\ls_i, \invals_i\rangle \tto \langle \gs_{i+1},\ls_{i+1}, \invals_{i+1}\rangle$.

\subsection{Accuracy and Extensions}
\label{subsec:accuracy}
Two main omissions from our semantics, when compared with the description in \autoref{subsec:runtime} are multiple data sources and signals. We briefly discuss these here.

\subsubsection*{Multiple Data Sources}
Because our semantics closely follows Vega's dataflow evaluation, adding multiple data sources does not introduce any difficulties. The only interesting situation would be transformations with multiple dependencies, which needs to avoid glitches~\cite{cooper-2006-dataflow}, i.e., a situation where a change to one dependency is propagated first, temporarily resulting in an inconsistent state. In Vega, merging of data from multiple sources is handled by \emph{materializing} the data (applying the changes to the source dataset), and \emph{collecting} the materialized data in a single array. This makes sure that all the data is sent out at once and avoids glitches.

\subsubsection*{Signals}
Signals, along with events, are a key concept of E-FRP, and they are the reason why it is hard to statically check Vega specifications. Indeed, they can be technically used anywhere except for identifiers, which must be determined statically.
Despite this, we can still allow them in specifications when they are used to a certain extent.

A signal is defined by a small number of expressions (init/update/bind and event handlers). When the expressions are all strings (plain parameters), this allows for shape inference of signals in the same way as for the $\formula$ transform, since we can just combine the expressions together appropriately.
Thus, signal definitions are not particularly problematic to analyze.

If we use signals to parameterize field or dataset names, we cannot guarantee correctness as the user can input any value into the signal at runtime. However, using signals in expressions poses no problems as long as we only need their shape to determine the shape of the expression.

\section{Properties}
\label{sec:sound}

Our key result is that the evaluation of well-typed \tinyvega specifications on well-formed input
never yields the $\error$ update. In the operational semantics, $\error$ is produced when
an expression evaluation $\Downarrow$ fails (in $\formula$ or $\filter$), when a~path
is not defined (in $\agg$), or when the data change is not well-formed, i.e., a data update references a nonexistent tuple (in $\agg$).

To show that this is the case, we first state the soundness of evaluation and show that
well-typed specification has a~unique derivation and a corresponding dataflow graph.
We then prove the usual progress and preservation lemmas.

\begin{assumption}[Evaluation soundness]
\label{asm:evalsound}
We assume $\Downarrow$ such that if $T \vdash e : T'$ then for all values $v$
(as defined in \autoref{fig:eval-objects}), if $\vdash v : T$ then
$e[\datum/v]\Downarrow v'$ and $\vdash v' : T'$.
\end{assumption}

\begin{lemma}[Uniqueness of typing derivation]
\label{lem:uniqueness}
If $\Gamma \vdash \spec : \Gamma'$ and $\Gamma \vdash \spec : \Gamma''$ then
$\Gamma' = \Gamma''$. Moreover the typing derivations for the two judgments are the same.
\end{lemma}
\begin{proof}
Typing rules are deterministic. The only non-syntax-directed rules
({\sc T-eArithNum}, {\sc T-eAddStrL}, {\sc T-eAddStrR}),
have non-overlapping assumptions.
\end{proof}

\begin{lemma}[Graph existence]
\label{lem:existence}
For any context $\Gamma = n_{\init}:T$, if
$\Gamma \vdash \spec : \Gamma'$ then the dependency graph
$(V, E), \lookup = \mathsf{build}(\{ n_{\init} \}, (d_1, \ldots, d_n))$
exists and is well-defined.
\end{lemma}
\begin{proof}
Graph construction is only undefined if $\lookup(n_\in)$ was undefined in $\buildD$, but
our type system ensures that datasets are defined before they are accessed.
\end{proof}

\noindent
Note that we only consider case with a single external data source. This is a simplification
without a loss of generality, because each of our transformations can only depend on a~single
data source.
Now, the two results mean that, for each transform $t$ in a well-typed specification, there
is a unique corresponding typing derivation $T \vdash t : T'$ that assigns input and output
types $T$ and $T'$ to the transform. Moreover, there is also a corresponding graph vertex:

\begin{lemma}[Correspondence]
  If $n_{\init}\!:\!T \vdash \spec\!:\!\Gamma'$ then
  $(V, E), \lookup = \mathsf{build}(\{ n_{\init} \}, \spec)$
  and for every vertex $v$, $v=(l,t) \in V$, there is a unique corresponding
  typing derivation $T \vdash t : T'$ for the transform $t$.
\end{lemma}
\begin{proof}
  Follows from Lemma~\ref{lem:uniqueness} and Lemma~\ref{lem:existence}.
\end{proof}

%

\noindent
The proof of the main soundness theorem will proceed by induction over $\rightarrow$.
The key step is to show that the semantics of individual transformation operations
turns well-typed upstream data changes into well-typed downstream data changes. The typing
of data changes is defined in \autoref{fig:update-types} and types the data passed
in the changes. Note that the $\error$ change is not typable. For stateful
transformations ($\cross$, $\agg$) we also need to ensure that values stored in the state are well-typed and that the received data changes are \textit{well-formed}. 

\begin{figure}
$$
\inferrule{\vdash v : T}{\vdash \add(v) : T}
\quad \inferrule{\vdash v : T}{\vdash \remove(v) : T}
\quad \inferrule{\vdash v : T\quad \vdash v' : T}{\vdash \update(v,v') : T}
$$ 
\caption{Typing of data updates}
\label{fig:update-types}
\end{figure}


\newcommand{\build}{\mathsf{build}}

\newcommand{\sep}{\mathbin{|}}

\begin{definition}[Well-formed input]
    An array of data changes $A = [c_1,\ldots,c_n]$ is said to be well-formed if $$\langle \gs \sep \ls, \invals\rangle \tto^n \langle \gs_n \sep \ls_n, \invals_n\rangle$$ for a dataflow graph $G$ with a single $\inputt$ vertex $v_\init$ with its changeset $\gs(v_\init) = A$ (otherwise, $\gs(v) = []$), and $\ls(v) = []$. The graph $G$ along with $\langle \gs \sep \ls, \invals\rangle$  is then referred to as a graph with well-formed input.
\end{definition}

\begin{definition}[Well-typed state]
Let $\mathrm{cod}(f)$ denote the co\-domain of a function $f$. Given a well-typed $\spec$, for each dataflow graph vertex $v=(l,t)$ with its unique corresponding
typing derivation $T \vdash t\!:\!T'$, we say that a~$\state$ is well-typed, written $\vdash_t \state~\mathsf{ok}$:
\begin{itemize}
\item[--] If $t = \cross$ then $\vdash_t \state~\mathsf{ok}$ if $\forall v \in \state\,.\,\vdash v : T$.
\item[--] If $t = \agg$ then $\vdash_t\state~\mathsf{ok}$ if $\forall v \in \mathrm{cod}(\state)\,.\,\vdash v : T'$.
\end{itemize}

\end{definition}

\begin{lemma}[Typing progress]
\label{lem:progress}
Let $n_\init\!:\!T \vdash \spec\!:\!\Gamma'$, and $(V, E), \lookup = \mathsf{build}({n_\init}, \spec)$, and $\langle\gs \sep \ls, \invals\rangle$ be a graph with well-formed input. If $\langle\gs \sep \ls, \invals\rangle \tto^* \langle\gs' \sep \ls', \invals'\rangle$,
then either $\gs' = \emptyset$ or there exists $\gs'', \ls''$ and $\invals''$ such that $\langle\gs', \ls', \invals'\rangle \tto \langle \gs'' \sep \ls'', \invals''\rangle$.
\end{lemma}
\begin{proof}
By induction on the length $k$ of $\langle\gs \sep \ls, \invals\rangle \tto^k \langle\gs' \sep \ls', \invals'\rangle$.
\begin{itemize}
\item If $k$ is less than the length of well-formed input, the well-formedness guarantees the existence of the next step.
\item Otherwise, assume that $\gs'$ of the first $k \ge 1$ vertices (w.r.t. topological order) is empty. Both vertices correspond to the transformations in \autoref{fig:core-transforms}. This means that we are currently processing changes coming from the $k$-th vertex to the $(k+1)$-th vertex. We can check by case analysis that if there was a bad change that would lead to \uv{$\not{\tto}$} during evaluation (e.g. $\add(t)$ for an already observed vertex), we could trace the problematic change to the predecessor of the $k$-th vertex.
\end{itemize}
\end{proof}

\begin{lemma}[Typing preservation]
\label{lem:preservation}
Assume a well-typed $\spec$ with a well-formed input, $\langle \gs \sep \ls, \invals\rangle \tto^* \langle \gs' \sep \ls', \invals'\rangle$,
 and $\langle \gs' \sep \ls', \invals'\rangle \tto \langle \gs'' \sep \ls'', \invals''\rangle$ at the vertex $v = (l, t)$ with $T \vdash t : T'$, $\vdash_t \state~\mathsf{ok}$, and the propagated change $u \in 
 \gs'(v)$, $\vdash u : T$ resulting in $\llbracket t\rrbracket_{\state}(u) = X, \state'$.
  Then $\vdash_t \state'~\mathsf{ok}$ and $\forall u' \in X$ it holds that $\vdash u' : T'$, and consequently $\error \not\in X$.


\end{lemma}

\begin{proof} By case analysis on the transformation $t$.
\end{proof}

The soundness for \tinyvega follows from progress and preservation in the standard way and states that, for any well-typed specification,
we can construct a dataflow graph and, for any updates sent to the graph, the graph-based
streaming evaluation will never emit the $\error$ update.

%
%

\section{Vega Specification Checker}
\label{sec:checker}
We have developed a proof-of-concept tool for checking invalid field access based on our type system. Besides following the type system, the tool also takes Vega's unexpected defaults into account.
The tool, along with examples of interesting invalid specifications, is available online at \url{https://github.com/kalivtrope/epsilon-lyrae/}.

\section{Related work}
\label{sec:related}

\subsubsection*{Vega Design and Implementation}
Vega has been primarily described from the user-centric perspective
\cite{satyanarayan-2014-vega} and the work has inspired a range of follow-up libraries
that refine or simplify the Vega programming model including Altair \cite{VanderPlas2018}
and Vega Lite \cite{vega-lite-2017}. Reactive Vega \cite{reactive-vega-architecture-2016}
describes the underlying streaming architecture, but a formal model of the semantics of
Vega has not been previously described.
Systems based on functional composition  \cite{petricek-2021-compost} may be more amenable to strong type checking.

\subsubsection*{Visualization grammars in practice}
The use of Vega-like libraries in practice has been subject
to research. Pu and Kay \cite{pu-2023-practice} point out that analysts sometimes make
``hard-to-evaluate silent errors'' also known as mirages \cite{mcnutt-2020-mirages}. Our type
system does not check for semantic errors, but may be able to prevent silent
errors of a more technical nature.
Linters \cite{chen2021,mcnutt2018linting} aim to recognize more
conceptual design errors, albeit heuristically, while Draco \cite{draco-2019} aims to
limit the surface area for errors by using high-level constraints.

Profilers and debuggers also have their equivalents
for data visualization. VegaProf \cite{yang-2023-profiling} instruments Vega runtime to
collect performance information---and exploring the performance implications of various design
choices discussed in our semantics would be an interesting use of such tool.
Debuggers have focused on helping the user understand data visualizations visually
\cite{vega-debugging-2016}; similar tools have been implemented using program slicing
and tracing techniques \cite{perera-2022-galois}.

Lu et al. \cite{lu-2025-bugs} classify bugs in the implementations
of data visualization libraries. A number of documented bugs could be prevented with the help
of precise operational semantics.

\subsubsection*{Functional reactive programming}
Vega is reactive and its transformations are akin to functional data processing.
Research on functional reactive programming thus provides valuable
references for defining the semantics of Vega.

Both the original FRP \cite{elliott-1997-fran} and systems based on arrows
\cite{wan-2000-frp,courtney-2003-yampa} are implemented using
dataflow graphs, although arrows make this more explicit.
The dataflow graph in our Vega semantics is static and so our system avoids problems
with memory leaks in the context of FRP \cite{krishnaswami-2013-leaks,bahr-2022-leaks}.

The semantics of FRP system has been described in multiple ways,
including a direct denotational style \cite{elliott-2009-pushpull},
using process categories \cite{jeltsch-2014-categorical} and ultrametric spaces
\cite{krishnaswami-2011-ultrametric}.

Our approach is closer to the graph-based operational semantics defined by Oeyen et
al.~\cite{oeyen-2023-graph}, which is based on two kinds of reduction rules---one kind for
rewiring the dependency graph and another for propagating values. In our model, the graph
is static and so we separate those more explicitly into graph construction and graph evaluation.

The semantics of arrow-based functional programming has also recently been formalized
in Rocq \cite{ischard-2025-frp} 
on arrows.

\subsubsection*{Adaptive and incremental computation}
Finally, the Vega runtime can also be linked to self-adjusting, adaptive and incremental
computation. In self-adjusting computation \cite{acar-2005-self}, the result is recomputed
efficiently when an input changes without recomputing the whole program. The implementation
is based on graphs and the change-propagation algorithm \cite{acar-2002-adaptive} is akin
to our graph evaluation (\autoref{sec:semantics-grapheval}).
The key difference is that Vega only works with datasets, rather than arbitrary data structures
and that the graph structure (in our model) remains static.

More generally, a range of incremental and adaptive systems use a form of dependency graph
to propagate updates to inputs. These include distributed systems \cite{murray-2016-timely},
and systems where propagation is on-demand \cite{hammer-2014-adapton,perera-2005-delta}.
The need for precise semantics of such systems has been shown by recent work on
Rocq formalization \cite{kumar-2025-incremental}, which also shows the correctness of a number
of optimizations.

\section{Conclusion}
Traditionally, formal semantics has been used to describe the behavior of elegant
programming systems. Recent works formalizing
the behavior of x86 processors~\cite{sewell-2010-x86} or React~\cite{Madsen2020ASF,lee-2025-reacttrace} show the importance of formally
studying real-world messy systems. Our work is a contribution to this effort. 

We presented a formal model of
data transformations in the widely used Vega library. Our calculus
accurately models evaluation as propagation of data updates through a dataflow
graph. We then described a sound type system that can prevent invalid field accesses in Vega
specifications. Although this is a seemingly simple error, it can have surprising
consequences for Vega visualizations yielding, for example, empty visualizations and
confusing warnings. 

\paragraph{Acknowledgments.}
We thank Jiří Beneš, jury members of the Student Research Competition at \guilsinglleft Programming\guilsinglright{} 2025,
as well as the anonymous referees of the paper for their valuable feedback. The work has been supported by the Charles University
grant PRIMUS/24/SCI/021 and by the Czech Ministry of Education, Youth and Sports grant ERC-CZ LL2325.

\balance
\bibliography{main}


\end{document}